\documentclass[11pt]{article}
\usepackage[a4paper]{geometry}
\usepackage{amsfonts, amsmath, amssymb, amsthm, graphicx, caption, authblk, multirow, makecell, framed, float, xcolor, enumitem, tikz, hyperref}
\theoremstyle{definition}

\newtheorem{lemma}{Lemma}

\theoremstyle{remark}

\definecolor{blk}{RGB}{63,63,63}
\newcommand*{\mybox}[1]{%
  \framebox{\raisebox{0cm}[0.5\baselineskip][0.05\baselineskip]{%
    \hbox to 0.10cm {\hss#1\hss}}}\hspace{0.05cm}}

\begin{document}
\title{Printing Protocol: Physical ZKPs for Decomposition Puzzles\thanks{A preliminary version of this paper \cite{fivecells} has appeared at LATINCRYPT 2023.}}
\author[1]{Suthee Ruangwises\thanks{\texttt{ruangwises@uec.ac.jp}}}
\author[1]{Mitsugu Iwamoto\thanks{\texttt{mitsugu@uec.ac.jp}}}
\affil[1]{Department of Informatics, The University of Electro-Communications, Tokyo, Japan}
\date{}
\maketitle

\begin{abstract}
Decomposition puzzles are pencil-and-paper logic puzzles that involve partitioning a rectangular grid into several regions to satisfy certain rules. In this paper, we construct a generic card-based protocol called \textit{printing protocol}, which can be used to physically verify solutions of decompositon puzzles. We apply the printing protocol to develop card-based zero-knowledge proof protocols for two such puzzles: Five Cells and Meadows. These protocols allow a prover to physically show that he/she knows solutions of the puzzles without revealing them.

\end{abstract}

\section{Introduction}
Pencil puzzles are puzzles that can be solved by writing down the solution on a paper. They must be solved using only logical reasoning and do not require additional knowledge. Examples of pencil puzzles include Sudoku, Nonogram, Kakuro, and Slitherlink.

Pencil puzzles can be categorized into several types based on their core theme. One of the common themes is to partition a rectangular grid into several regions to satisfy certain rules. We call these puzzles \textit{decomposition puzzles}.

\subsection{Five Cells}
\textit{Five Cells} is a decomposition puzzle created by Nikoli, a Japanese publisher famous for developing many popular pencil puzzles. The puzzle consists of an $m \times n$ rectangular grid, with some cells containing a number. The player has to partition the grid into pentominoes. The number in each cell indicates the number of edges of that cell that are borders of pentominoes (including the outer boundary of the grid) \cite{nikoli}. See Fig. \ref{fig0}.

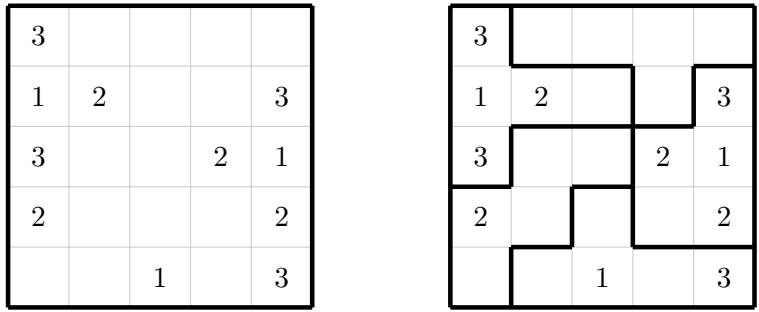
\begin{figure}
\centering
\begin{tikzpicture}
\draw[step=0.8cm,color={rgb:black,1;white,4}] (0,0) grid (4,4);

\draw[line width=0.6mm] (0,0) -- (0,4);
\draw[line width=0.6mm] (4,0) -- (4,4);
\draw[line width=0.6mm] (0,0) -- (4,0);
\draw[line width=0.6mm] (0,4) -- (4,4);

\node at (0.4,1.2) {2};
\node at (0.4,2) {3};
\node at (0.4,2.8) {1};
\node at (0.4,3.6) {3};
\node at (1.2,2.8) {2};
\node at (2,0.4) {1};
\node at (2.8,2) {2};
\node at (3.6,0.4) {3};
\node at (3.6,1.2) {2};
\node at (3.6,2) {1};
\node at (3.6,2.8) {3};
\end{tikzpicture}
\hspace{1.5cm}
\begin{tikzpicture}
\draw[step=0.8cm,color={rgb:black,1;white,4}] (0,0) grid (4,4);

\draw[line width=0.6mm] (0,0) -- (0,4);
\draw[line width=0.6mm] (0.8,0) -- (0.8,0.8);
\draw[line width=0.6mm] (0.8,1.6) -- (0.8,2.4);
\draw[line width=0.6mm] (0.8,3.2) -- (0.8,4);
\draw[line width=0.6mm] (1.6,0.8) -- (1.6,1.6);
\draw[line width=0.6mm] (2.4,0.8) -- (2.4,3.2);
\draw[line width=0.6mm] (3.2,2.4) -- (3.2,3.2);
\draw[line width=0.6mm] (4,0) -- (4,4);
\draw[line width=0.6mm] (0,0) -- (4,0);
\draw[line width=0.6mm] (0.8,0.8) -- (1.6,0.8);
\draw[line width=0.6mm] (2.4,0.8) -- (4,0.8);
\draw[line width=0.6mm] (0,1.6) -- (0.8,1.6);
\draw[line width=0.6mm] (1.6,1.6) -- (2.4,1.6);
\draw[line width=0.6mm] (0.8,2.4) -- (3.2,2.4);
\draw[line width=0.6mm] (0.8,3.2) -- (2.4,3.2);
\draw[line width=0.6mm] (3.2,3.2) -- (4,3.2);
\draw[line width=0.6mm] (0,4) -- (4,4);

\node at (0.4,1.2) {2};
\node at (0.4,2) {3};
\node at (0.4,2.8) {1};
\node at (0.4,3.6) {3};
\node at (1.2,2.8) {2};
\node at (2,0.4) {1};
\node at (2.8,2) {2};
\node at (3.6,0.4) {3};
\node at (3.6,1.2) {2};
\node at (3.6,2) {1};
\node at (3.6,2.8) {3};
\end{tikzpicture}
\caption{An example of a $5 \times 5$ Five Cells puzzle (left) and its solution (right)}
\label{fig0}
\end{figure}

Deciding solvability of a given Five Cells puzzle is NP-complete \cite{np}.

\subsection{Meadows}
\textit{Meadows} is another decomposition puzzle consisting of an $n \times n$ square grid, with some cells containing a dot. The player has to partition the grid into squares such that each square contains exactly one dotted cell. See Fig. \ref{fig00}.

The rules of Meadows are similar to those of \textit{Shikaku}\footnote{Although there exists a card-based ZKP for Shikaku \cite{shikaku}, the protocol uses a technique specifically designed for Shikaku, which cannot be applied to Meadows.}, a more well-known decomposition puzzle developed by Nikoli, so the puzzle may be considered a variant of Shikaku \cite{janko}, which is also NP-complete \cite{np2}.

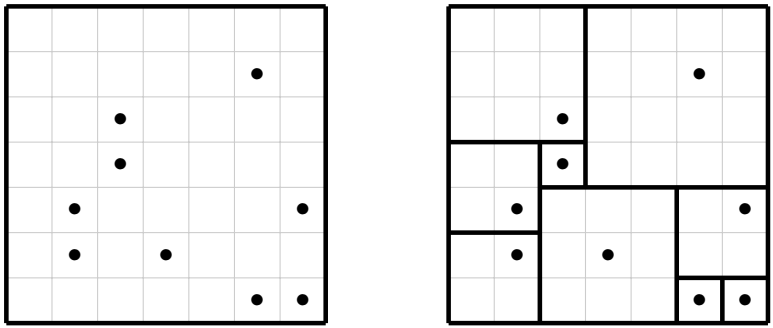
\begin{figure}
\centering
\begin{tikzpicture}
\draw[step=0.6cm,color={rgb:black,1;white,4}] (0,0) grid (4.2,4.2);

\draw[line width=0.6mm] (0,0) -- (0,4.2);
\draw[line width=0.6mm] (4.2,0) -- (4.2,4.2);
\draw[line width=0.6mm] (0,0) -- (4.2,0);
\draw[line width=0.6mm] (0,4.2) -- (4.2,4.2);

\node at (0.9,0.9) {$\bullet$};
\node at (0.9,1.5) {$\bullet$};
\node at (1.5,2.1) {$\bullet$};
\node at (1.5,2.7) {$\bullet$};
\node at (2.1,0.9) {$\bullet$};
\node at (3.3,0.3) {$\bullet$};
\node at (3.3,3.3) {$\bullet$};
\node at (3.9,0.3) {$\bullet$};
\node at (3.9,1.5) {$\bullet$};
\end{tikzpicture}
\hspace{1.3cm}
\begin{tikzpicture}
\draw[step=0.6cm,color={rgb:black,1;white,4}] (0,0) grid (4.2,4.2);

\draw[line width=0.6mm] (0,0) -- (0,4.2);
\draw[line width=0.6mm] (4.2,0) -- (4.2,4.2);
\draw[line width=0.6mm] (0,0) -- (4.2,0);
\draw[line width=0.6mm] (0,4.2) -- (4.2,4.2);

\node at (0.9,0.9) {$\bullet$};
\node at (0.9,1.5) {$\bullet$};
\node at (1.5,2.1) {$\bullet$};
\node at (1.5,2.7) {$\bullet$};
\node at (2.1,0.9) {$\bullet$};
\node at (3.3,0.3) {$\bullet$};
\node at (3.3,3.3) {$\bullet$};
\node at (3.9,0.3) {$\bullet$};
\node at (3.9,1.5) {$\bullet$};

\draw[line width=0.6mm] (1.2,0) -- (1.2,2.4);
\draw[line width=0.6mm] (1.8,1.8) -- (1.8,4.2);
\draw[line width=0.6mm] (3,0) -- (3,1.8);
\draw[line width=0.6mm] (3.6,0) -- (3.6,0.6);
\draw[line width=0.6mm] (3,0.6) -- (4.2,0.6);
\draw[line width=0.6mm] (0,1.2) -- (1.2,1.2);
\draw[line width=0.6mm] (1.2,1.8) -- (4.2,1.8);
\draw[line width=0.6mm] (0,2.4) -- (1.8,2.4);
\end{tikzpicture}
\caption{An example of a $7 \times 7$ Meadows puzzle (left) and its solution (right)}
\label{fig00}
\end{figure}

\subsection{Zero-Knowledge Proof}
Suppose that Panthalassa, a puzzle expert, creates a pencil puzzle and challenges his friend Vodka to solve it. He wants to convince her that the puzzle indeed has a solution, but does not want to reveal the solution itself. A \textit{zero-knowledge proof (ZKP)} makes this seemingly difficult task possible.

Formally, a ZKP is an interactive protocol between a prover $P$ and a verifier $V$, where both are given a computational problem $x$, but only $P$ knows its solution $w$. Also, the computational power of $V$ is limited, so $V$ cannot obtain $w$ from $x$. A ZKP has to satisfy the following three properties.

\begin{enumerate}
	\item \textbf{Completeness:} If $P$ knows $w$, then $V$ accepts with high probability. (In this paper, we consider the \textit{perfect completeness} property where $V$ always accepts.)
	\item \textbf{Soundness:} If $P$ does not know $w$, then $V$ rejects with high probability. (In this paper, we consider the \textit{perfect soundness} property where $V$ always rejects.)
	\item \textbf{Zero-knowledge:} $V$ learns nothing about $w$. Formally, there exists a probabilistic polynomial time algorithm $S$ (called a \textit{simulator}) that does not know $w$ but has an access to $V$, and the outputs of $S$ follow the same probability distribution as the ones of the real protocol.
\end{enumerate}

The concept of a ZKP was first introduced in 1989 by Goldwasser et al. \cite{zkp0}. It has been proved that every NP problem has a ZKP \cite{zkp}, so a computational ZKP for any NP puzzle can be constructed via a reduction to another problem. Such construction, however, is unintuitive and looks unconvincing. Therefore, many researchers instead focused on constructing physical ZKPs using a deck of playing cards. These card-based protocols have benefits that they require only portable objects easily found in everyday life and do not require computers. They are also easy to understand and verify the correctness and security, even for non-experts, and thus can be used for didactic purpose.

Card-based ZKP protocols for many pencil puzzles have been developed\footnote{Among these puzzles, only Shikaku is a decomposition puzzle, and its ZKP protocol \cite{shikaku} uses a completely different approach from the one in this paper.}, including Sudoku \cite{sudoku0,sudoku2,sudoku}, Nonogram \cite{nonogram,nonogram2}, Akari \cite{akari}, Takuzu \cite{akari,takuzu}, Kakuro \cite{akari,kakuro}, KenKen \cite{akari}, Makaro \cite{makaro,makaro2}, Norinori \cite{norinori}, Slitherlink \cite{slitherlink}, Juosan \cite{takuzu}, Numberlink \cite{numberlink}, Suguru \cite{suguru}, Ripple Effect \cite{ripple}, Nurikabe \cite{nurikabe}, Hitori \cite{nurikabe}, Bridges \cite{bridges}, Masyu \cite{slitherlink}, Heyawake \cite{nurikabe}, Shikaku \cite{shikaku}, Usowan \cite{usowan}, Nurimisaki \cite{nurimisaki}, ABC End View \cite{abc,goishi}, Goishi Hiroi \cite{goishi}, Moon-or-Sun \cite{moon}, and Kurodoko \cite{nurimisaki}, as well as those for non-pencil logic puzzles such as Cryptarithmetic \cite{crypta} and Ball sort puzzle \cite{ball}.

\subsection{Our Contribution}
In this paper, we construct a generic card-based protocol called \textit{printing protocol}. This protocol prints some numbers from the template onto a target area from the grid, and also verifies that the printed numbers do not overlap with the already existing numbers in that area. The printing protocol can be used to verify solutions of many decomposition puzzles. We show two applications of it by developing ZKP protocols for two such puzzles: Five Cells and Meadows.

\section{Preliminaries}
\subsection{Cards} \label{encoding}
Each card used in our protocol either has an integer written on the front side (e.g. \hbox{\mybox{1},} \mybox{2}) or is a blank card with nothing written on the front side (\mybox{}). All cards have indistinguishable back sides denoted by \mybox{?}.

\subsection{Pile-Shifting Shuffle}
Given a matrix $M$ of cards, a \textit{pile-shifting shuffle} \cite{polygon} shifts the columns of $M$ by a uniformly random cyclic shift unknown to all parties (see Fig. \ref{fig1}). It can be implemented in real world by putting all cards in each column into an envelope, and letting all parties take turns to apply \textit{Hindu cuts} (taking some envelopes from the bottom and putting them on the top) to the pile of envelopes.

Note that each card in $M$ may be replaced by a stack of cards, and the protocol still works in the same way as long as every stack in the same row has the same number of cards.

\begin{figure}
\centering
\begin{tikzpicture}
\node at (0,0) {\mybox{?}};
\node at (0.5,0) {\mybox{?}};
\node at (1,0) {\mybox{?}};
\node at (1.5,0) {\mybox{?}};
\node at (2,0) {\mybox{?}};
\node at (2.5,0) {\mybox{?}};

\node at (0,0.6) {\mybox{?}};
\node at (0.5,0.6) {\mybox{?}};
\node at (1,0.6) {\mybox{?}};
\node at (1.5,0.6) {\mybox{?}};
\node at (2,0.6) {\mybox{?}};
\node at (2.5,0.6) {\mybox{?}};

\node at (0,1.2) {\mybox{?}};
\node at (0.5,1.2) {\mybox{?}};
\node at (1,1.2) {\mybox{?}};
\node at (1.5,1.2) {\mybox{?}};
\node at (2,1.2) {\mybox{?}};
\node at (2.5,1.2) {\mybox{?}};

\node at (0,1.8) {\mybox{?}};
\node at (0.5,1.8) {\mybox{?}};
\node at (1,1.8) {\mybox{?}};
\node at (1.5,1.8) {\mybox{?}};
\node at (2,1.8) {\mybox{?}};
\node at (2.5,1.8) {\mybox{?}};

\node at (0,2.4) {\mybox{?}};
\node at (0.5,2.4) {\mybox{?}};
\node at (1,2.4) {\mybox{?}};
\node at (1.5,2.4) {\mybox{?}};
\node at (2,2.4) {\mybox{?}};
\node at (2.5,2.4) {\mybox{?}};

\node at (-0.4,0) {5};
\node at (-0.4,0.6) {4};
\node at (-0.4,1.2) {3};
\node at (-0.4,1.8) {2};
\node at (-0.4,2.4) {1};

\node at (0,2.9) {1};
\node at (0.5,2.9) {2};
\node at (1,2.9) {3};
\node at (1.5,2.9) {4};
\node at (2,2.9) {5};
\node at (2.5,2.9) {6};

\node at (3.4,1.2) {\LARGE{$\Rightarrow$}};
\end{tikzpicture}
\begin{tikzpicture}
\node at (0,0) {\mybox{?}};
\node at (0.5,0) {\mybox{?}};
\node at (1,0) {\mybox{?}};
\node at (1.5,0) {\mybox{?}};
\node at (2,0) {\mybox{?}};
\node at (2.5,0) {\mybox{?}};

\node at (0,0.6) {\mybox{?}};
\node at (0.5,0.6) {\mybox{?}};
\node at (1,0.6) {\mybox{?}};
\node at (1.5,0.6) {\mybox{?}};
\node at (2,0.6) {\mybox{?}};
\node at (2.5,0.6) {\mybox{?}};

\node at (0,1.2) {\mybox{?}};
\node at (0.5,1.2) {\mybox{?}};
\node at (1,1.2) {\mybox{?}};
\node at (1.5,1.2) {\mybox{?}};
\node at (2,1.2) {\mybox{?}};
\node at (2.5,1.2) {\mybox{?}};

\node at (0,1.8) {\mybox{?}};
\node at (0.5,1.8) {\mybox{?}};
\node at (1,1.8) {\mybox{?}};
\node at (1.5,1.8) {\mybox{?}};
\node at (2,1.8) {\mybox{?}};
\node at (2.5,1.8) {\mybox{?}};

\node at (0,2.4) {\mybox{?}};
\node at (0.5,2.4) {\mybox{?}};
\node at (1,2.4) {\mybox{?}};
\node at (1.5,2.4) {\mybox{?}};
\node at (2,2.4) {\mybox{?}};
\node at (2.5,2.4) {\mybox{?}};

\node at (-0.4,0) {5};
\node at (-0.4,0.6) {4};
\node at (-0.4,1.2) {3};
\node at (-0.4,1.8) {2};
\node at (-0.4,2.4) {1};

\node at (0,2.9) {5};
\node at (0.5,2.9) {6};
\node at (1,2.9) {1};
\node at (1.5,2.9) {2};
\node at (2,2.9) {3};
\node at (2.5,2.9) {4};
\end{tikzpicture}
\caption{A pile-shifting shuffle on a $5 \times 6$ matrix}
\label{fig1}
\end{figure}
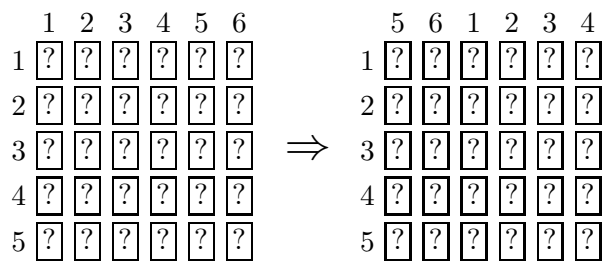

\subsection{Chosen Cut Protocol} \label{chosen}
Given a sequence $C = (c_1,c_2,...,c_q)$ of $q$ face-down cards, a \textit{chosen cut protocol} \cite{koch} allows $P$ to select a desired card $c_i$ to use in other protocols without revealing $i$ to $V$. Optionally, the protocol can also revert $C$ back to its original state after $P$ finishes using $c_i$.

\begin{figure}
\centering
\begin{tikzpicture}
\node at (0.0,2.4) {\mybox{?}};
\node at (0.6,2.4) {\mybox{?}};
\node at (1.2,2.4) {...};
\node at (1.8,2.4) {\mybox{?}};
\node at (2.4,2.4) {\mybox{?}};
\node at (3.0,2.4) {\mybox{?}};
\node at (3.6,2.4) {...};
\node at (4.2,2.4) {\mybox{?}};

\node at (0.0,2) {$c_1$};
\node at (0.6,2) {$c_2$};
\node at (1.8,2) {$c_{i-1}$};
\node at (2.4,2) {$c_i$};
\node at (3.0,2) {$c_{i+1}$};
\node at (4.2,2) {$c_q$};

\node at (0.0,1.4) {\mybox{?}};
\node at (0.6,1.4) {\mybox{?}};
\node at (1.2,1.4) {...};
\node at (1.8,1.4) {\mybox{?}};
\node at (2.4,1.4) {\mybox{?}};
\node at (3.0,1.4) {\mybox{?}};
\node at (3.6,1.4) {...};
\node at (4.2,1.4) {\mybox{?}};

\node at (0.0,1) {0};
\node at (0.6,1) {0};
\node at (1.8,1) {0};
\node at (2.4,1) {1};
\node at (3.0,1) {0};
\node at (4.2,1) {0};

\node at (0.0,0.4) {\mybox{1}};
\node at (0.6,0.4) {\mybox{0}};
\node at (1.2,0.4) {...};
\node at (1.8,0.4) {\mybox{0}};
\node at (2.4,0.4) {\mybox{0}};
\node at (3.0,0.4) {\mybox{0}};
\node at (3.6,0.4) {...};
\node at (4.2,0.4) {\mybox{0}};
\end{tikzpicture}
\caption{A $3 \times q$ matrix $M$ constructed in Step 1 of the chosen cut protocol}
\label{fig2}
\end{figure}
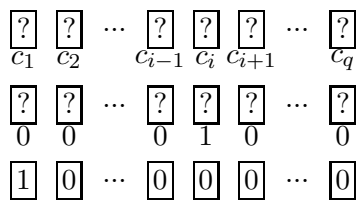

$P$ performs the following steps (called the \textit{first part}).
\begin{enumerate}
	\item Construct the following $3 \times q$ matrix $M$ (see Fig. \ref{fig2}).
	\begin{enumerate}
		\item In Row 1, place the sequence $C$.
		\item In Row 2, place a face-down \mybox{1} at Column $i$ and a face-down \mybox{0} at each other column.
		\item In Row 3, place a face-up \mybox{1} at Column 1 and a face-up \mybox{0} at each other column.
	\end{enumerate}
	\item Turn over all face-up cards. Apply the pile-shifting shuffle to $M$.
	\item Turn over all cards in Row 2 and locate the position of the only \mybox{1}. A card in Row 1 directly above this \mybox{1} will be the card $c_i$ as desired.
\end{enumerate}

Optionally, if $P$ wants to revert $C$ back to its original state, $P$ continues the following steps (called the \textit{second part}) after finishing using $c_i$ in other protocols.

\begin{enumerate}
	\setcounter{enumi}{3}
	\item Place $c_i$ back into $M$ at the same position.
	\item Turn over all face-up cards. Apply the pile-shifting shuffle to $M$.
	\item Turn over all cards in Row 3 and locate the position of the only \mybox{1}. Shift the columns of $M$ cyclically such that this \mybox{1} moves to Column 1. $M$ is now reverted back to its original state.
\end{enumerate}

Note that each card in $C$ may be replaced by a stack of cards, and the protocol still works in the same way as long as every stack has the same number of cards.

\section{Printing Protocol} \label{print}
Given a $p \times q$ matrix of cards called a \textit{template} (which contains some non-blank cards, and possibly some blank cards) and another $p \times q$ matrix of cards representing an area from the puzzle grid, all known to $P$ but not to $V$. A printing protocol verifies that positions in the area corresponding to non-blank cards in the template are initially empty (consisting of all blank cards). The protocol then places all non-blank cards from the template at the corresponding positions in the area, replacing the original blank cards (see Fig. \ref{fig3}) without revealing any card to $V$.

\begin{figure}
\centering
\begin{tikzpicture}
\node at (0,0) {\mybox{}};
\node at (0.5,0) {\mybox{}};
\node at (1,0) {\mybox{}};
\node at (1.5,0) {\mybox{}};
\node at (2,0) {\mybox{}};

\node at (0,0.6) {\mybox{}};
\node at (0.5,0.6) {\mybox{1}};
\node at (1,0.6) {\mybox{2}};
\node at (1.5,0.6) {\mybox{3}};
\node at (2,0.6) {\mybox{}};

\node at (0,1.2) {\mybox{}};
\node at (0.5,1.2) {\mybox{}};
\node at (1,1.2) {\mybox{}};
\node at (1.5,1.2) {\mybox{4}};
\node at (2,1.2) {\mybox{}};

\node at (0,1.8) {\mybox{}};
\node at (0.5,1.8) {\mybox{}};
\node at (1,1.8) {\mybox{}};
\node at (1.5,1.8) {\mybox{}};
\node at (2,1.8) {\mybox{}};

\node at (1,-0.5) {Template};

\node at (2.7,0.9) {\LARGE{+}};
\end{tikzpicture}
\begin{tikzpicture}
\node at (0,0) {\mybox{}};
\node at (0.5,0) {\mybox{}};
\node at (1,0) {\mybox{0}};
\node at (1.5,0) {\mybox{}};
\node at (2,0) {\mybox{}};

\node at (0,0.6) {\mybox{}};
\node at (0.5,0.6) {\mybox{}};
\node at (1,0.6) {\mybox{}};
\node at (1.5,0.6) {\mybox{}};
\node at (2,0.6) {\mybox{1}};

\node at (0,1.2) {\mybox{0}};
\node at (0.5,1.2) {\mybox{}};
\node at (1,1.2) {\mybox{5}};
\node at (1.5,1.2) {\mybox{}};
\node at (2,1.2) {\mybox{2}};

\node at (0,1.8) {\mybox{}};
\node at (0.5,1.8) {\mybox{1}};
\node at (1,1.8) {\mybox{}};
\node at (1.5,1.8) {\mybox{3}};
\node at (2,1.8) {\mybox{}};

\node at (1,-0.5) {Area};

\node at (2.8,0.9) {\LARGE{$\Rightarrow$}};
\end{tikzpicture}
\begin{tikzpicture}
\node at (0,0) {\mybox{}};
\node at (0.5,0) {\mybox{}};
\node at (1,0) {\mybox{0}};
\node at (1.5,0) {\mybox{}};
\node at (2,0) {\mybox{}};

\node at (0,0.6) {\mybox{}};
\node at (0.5,0.6) {\mybox{1}};
\node at (1,0.6) {\mybox{2}};
\node at (1.5,0.6) {\mybox{3}};
\node at (2,0.6) {\mybox{1}};

\node at (0,1.2) {\mybox{0}};
\node at (0.5,1.2) {\mybox{}};
\node at (1,1.2) {\mybox{5}};
\node at (1.5,1.2) {\mybox{4}};
\node at (2,1.2) {\mybox{2}};

\node at (0,1.8) {\mybox{}};
\node at (0.5,1.8) {\mybox{1}};
\node at (1,1.8) {\mybox{}};
\node at (1.5,1.8) {\mybox{3}};
\node at (2,1.8) {\mybox{}};

\node at (1,-0.5) {Area};
\end{tikzpicture}

\begin{tikzpicture}
\node at (0,0) {\mybox{}};
\node at (0.5,0) {\mybox{}};
\node at (1,0) {\mybox{}};
\node at (1.5,0) {\mybox{}};
\node at (2,0) {\mybox{}};

\node at (0,0.6) {\mybox{}};
\node at (0.5,0.6) {\mybox{1}};
\node at (1,0.6) {\mybox{2}};
\node at (1.5,0.6) {\mybox{3}};
\node at (2,0.6) {\mybox{}};

\node at (0,1.2) {\mybox{}};
\node at (0.5,1.2) {\mybox{}};
\node at (1,1.2) {\mybox{}};
\node at (1.5,1.2) {\mybox{4}};
\node at (2,1.2) {\mybox{}};

\node at (0,1.8) {\mybox{}};
\node at (0.5,1.8) {\mybox{}};
\node at (1,1.8) {\mybox{}};
\node at (1.5,1.8) {\mybox{}};
\node at (2,1.8) {\mybox{}};

\node at (1,-0.5) {Template};

\node at (2.7,0.9) {\LARGE{+}};
\end{tikzpicture}
\begin{tikzpicture}
\node at (0,0) {\mybox{}};
\node at (0.5,0) {\mybox{}};
\node at (1,0) {\mybox{0}};
\node at (1.5,0) {\mybox{}};
\node at (2,0) {\mybox{}};

\node at (0,0.6) {\mybox{}};
\node at (0.5,0.6) {\mybox{}};
\node at (1,0.6) {\mybox{1}};
\node at (1.5,0.6) {\mybox{}};
\node at (2,0.6) {\mybox{1}};

\node at (0,1.2) {\mybox{0}};
\node at (0.5,1.2) {\mybox{}};
\node at (1,1.2) {\mybox{5}};
\node at (1.5,1.2) {\mybox{}};
\node at (2,1.2) {\mybox{2}};

\node at (0,1.8) {\mybox{}};
\node at (0.5,1.8) {\mybox{1}};
\node at (1,1.8) {\mybox{}};
\node at (1.5,1.8) {\mybox{3}};
\node at (2,1.8) {\mybox{}};

\node at (1,-0.5) {Area};

\node at (2.8,0.9) {\LARGE{$\Rightarrow$}};
\end{tikzpicture}
\begin{tikzpicture}
\node at (-0.15,-0.65) {};
\node at (2.15,1.95) {};

\node at (1,0.9) {\LARGE{Reject}};
\end{tikzpicture}
\caption{Examples of a succesful printing of numbers from a $4 \times 5$ template onto a $4 \times 5$ area (top) and an unsuccessful printing because of overlapping numbers (bottom)}
\label{fig3}
\end{figure}
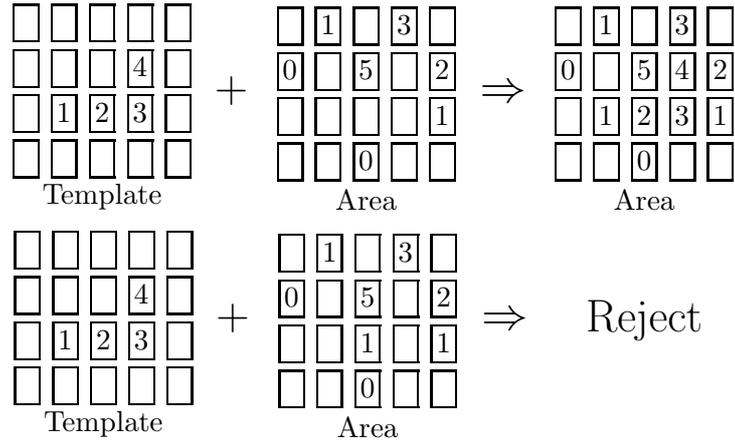

$P$ performs the following steps.
\begin{enumerate}
	\item Place each card from the template on top of each corresponding card from the area, creating $pq$ stacks of two cards.
	\item For each of the $pq$ stacks, perform the following steps.
	\begin{enumerate}
		\item Apply the first part of the chosen cut protocol to select a blank card. (If the preconditions are met, at least one card must be blank; if both cards are blank, $P$ can select any of them.)
		\item Reveal the selected card that it is a blank card (otherwise $V$ rejects) and remove it from the stack.
	\end{enumerate}
\end{enumerate}

After these steps, all non-blank cards from the template are placed at the corresponding positions in the area. $V$ is also convinced that these positions in the area were initially empty (consisting of all blank cards) before the protocol.

Using the printing protocol, $P$ can separately print each region in the partition onto the grid to construct $P$'s solution of a decomposition puzzle. $V$ will be convinced that the regions do not overlap with one another.

\section{ZKP Protocol for Five Cells} \label{fivecells}
First, from the solution of Five Cells, one can fill a number on every cell according to the rule of the puzzle (the number in each cell being the number of edges of that cell that are borders of pentominoes). We call this instance an \textit{extended solution} of the puzzle.

The key observation is that there are only $\Theta(1)$ different types of pentomino. Namely, there are 63 of them \cite{oeis}.\footnote{A pentomino obtained by rotating or reflecting another pentomino is considered a different one.} Furthermore, inside each type of pentomino in the extended solution, the number in each cell is fixed for that type. Therefore, we need to construct only 63 different templates.

In our protocol, $P$ constructs 63 templates, one for each type of pentomino. Each template has size $5 \times 5$ (as the height and length of a pentomino are at most five), and the pentomino is placed at the top-left corner of the template. A cell inside the pentomino is represented by a card with a number equal to the number of edges of that cell that are borders of the pentomino, while a cell outside the pentomino is represented by a blank card (see Fig. \ref{figA}).

\begin{figure}
\centering
\begin{tikzpicture}
\draw[color={rgb:black,1;white,4}] (0.8,0.8) -- (0.8,1.6);
\draw[color={rgb:black,1;white,4}] (1.6,0.8) -- (1.6,1.6);

\draw[color={rgb:black,1;white,4}] (0.8,0.8) -- (1.6,0.8);
\draw[color={rgb:black,1;white,4}] (0.8,1.6) -- (1.6,1.6);

\draw[line width=0.6mm] (0,0.8) -- (0,1.6);
\draw[line width=0.6mm] (0.8,0) -- (0.8,0.8);
\draw[line width=0.6mm] (0.8,1.6) -- (0.8,2.4);
\draw[line width=0.6mm] (1.6,0) -- (1.6,0.8);
\draw[line width=0.6mm] (1.6,1.6) -- (1.6,2.4);
\draw[line width=0.6mm] (2.4,0.8) -- (2.4,1.6);

\draw[line width=0.6mm] (0.8,0) -- (1.6,0);
\draw[line width=0.6mm] (0,0.8) -- (0.8,0.8);
\draw[line width=0.6mm] (1.6,0.8) -- (2.4,0.8);
\draw[line width=0.6mm] (0,1.6) -- (0.8,1.6);
\draw[line width=0.6mm] (1.6,1.6) -- (2.4,1.6);
\draw[line width=0.6mm] (0.8,2.4) -- (1.6,2.4);

\node at (0.4,1.2) {3};
\node at (1.2,0.4) {3};
\node at (1.2,1.2) {0};
\node at (1.2,2) {3};
\node at (2,1.2) {3};

\node at (0,1) {};
\node at (3,-0.2) {};

\node at (3.1,1.2) {\LARGE{$\Rightarrow$}};
\end{tikzpicture}
\begin{tikzpicture}
\node at (0,0) {\mybox{}};
\node at (0.5,0) {\mybox{}};
\node at (1,0) {\mybox{}};
\node at (1.5,0) {\mybox{}};
\node at (2,0) {\mybox{}};

\node at (0,0.6) {\mybox{}};
\node at (0.5,0.6) {\mybox{}};
\node at (1,0.6) {\mybox{}};
\node at (1.5,0.6) {\mybox{}};
\node at (2,0.6) {\mybox{}};

\node at (0,1.2) {\mybox{}};
\node at (0.5,1.2) {\mybox{3}};
\node at (1,1.2) {\mybox{}};
\node at (1.5,1.2) {\mybox{}};
\node at (2,1.2) {\mybox{}};

\node at (0,1.8) {\mybox{3}};
\node at (0.5,1.8) {\mybox{0}};
\node at (1,1.8) {\mybox{3}};
\node at (1.5,1.8) {\mybox{}};
\node at (2,1.8) {\mybox{}};

\node at (0,2.4) {\mybox{}};
\node at (0.5,2.4) {\mybox{3}};
\node at (1,2.4) {\mybox{}};
\node at (1.5,2.4) {\mybox{}};
\node at (2,2.4) {\mybox{}};
\end{tikzpicture}

\begin{tikzpicture}
\draw[color={rgb:black,1;white,4}] (1.2,0.8) -- (1.2,2.4);

\draw[color={rgb:black,1;white,4}] (0.4,0.8) -- (1.2,0.8);
\draw[color={rgb:black,1;white,4}] (0.4,1.6) -- (2,1.6);

\draw[line width=0.6mm] (0.4,0) -- (0.4,2.4);
\draw[line width=0.6mm] (1.2,0) -- (1.2,0.8);
\draw[line width=0.6mm] (2,0.8) -- (2,2.4);

\draw[line width=0.6mm] (0.4,0) -- (1.2,0);
\draw[line width=0.6mm] (1.2,0.8) -- (2,0.8);
\draw[line width=0.6mm] (0.4,2.4) -- (2,2.4);

\node at (0.8,0.4) {3};
\node at (0.8,1.2) {1};
\node at (0.8,2) {2};
\node at (1.6,1.2) {2};
\node at (1.6,2) {2};

\node at (0,3) {};
\node at (3,-0.2) {};

\node at (3.1,1.2) {\LARGE{$\Rightarrow$}};
\end{tikzpicture}
\begin{tikzpicture}
\node at (0,0) {\mybox{}};
\node at (0.5,0) {\mybox{}};
\node at (1,0) {\mybox{}};
\node at (1.5,0) {\mybox{}};
\node at (2,0) {\mybox{}};

\node at (0,0.6) {\mybox{}};
\node at (0.5,0.6) {\mybox{}};
\node at (1,0.6) {\mybox{}};
\node at (1.5,0.6) {\mybox{}};
\node at (2,0.6) {\mybox{}};

\node at (0,1.2) {\mybox{3}};
\node at (0.5,1.2) {\mybox{}};
\node at (1,1.2) {\mybox{}};
\node at (1.5,1.2) {\mybox{}};
\node at (2,1.2) {\mybox{}};

\node at (0,1.8) {\mybox{1}};
\node at (0.5,1.8) {\mybox{2}};
\node at (1,1.8) {\mybox{}};
\node at (1.5,1.8) {\mybox{}};
\node at (2,1.8) {\mybox{}};

\node at (0,2.4) {\mybox{2}};
\node at (0.5,2.4) {\mybox{2}};
\node at (1,2.4) {\mybox{}};
\node at (1.5,2.4) {\mybox{}};
\node at (2,2.4) {\mybox{}};

\node at (1,3) {};
\end{tikzpicture}
\caption{Templates of the X-shaped pentomino and the P-shaped pentomino}
\label{figA}
\end{figure}
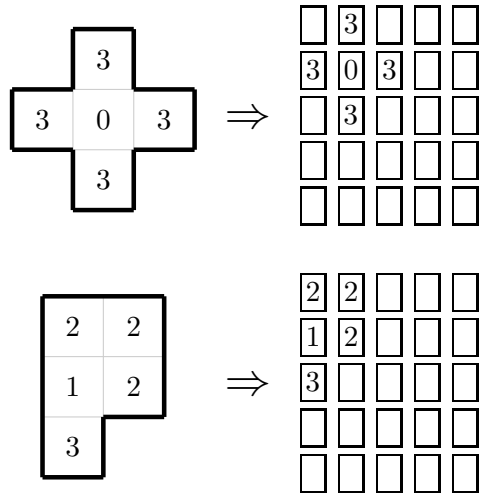

\subsection{Main Protocol}
Initially, $P$ publicly places a blank card on every cell in the grid. To handle edge cases, $P$ publicly appends four rows and four columns of ``dummy'' blank cards to the bottom and to the right of the grid. Then, $P$ turns all card face-down. We now have an $(m+4) \times (n+4)$ matrix of cards.

Observe that if we arrange all $(m+4)(n+4)$ cards in the matrix into a single sequence $A=(a_1,a_2,...,a_{(m+4)(n+4)})$, starting at the top-left corner and going from left to right in Row 1, then from left to right in Row 2, and so on, we can locate exactly where the four neighbors of any given card are. Namely, the cards on the neighbor to the left, right, top, and bottom of a cell containing $a_i$ are $a_{i-1}$, $a_{i+1}$, $a_{i-n-4}$, and $a_{i+n+4}$, respectively.

Also, $P$ constructs 63 templates of all 63 types of pentomino and lets $V$ verify that all templates are correct (otherwise $V$ rejects).

Suppose that in $P$'s extended solution, the grid is partitioned into $k$ pentominoes $B_1,B_2,...,B_k$, where $k=mn/5$. For each $i=1,2,...,k$, $P$ performs the following steps.

\begin{enumerate}
	\item Apply the first part of the chosen cut protocol to select a $5 \times 5$ area containing pentomino $B_i$ in the same position as in its corresponding template. (To be precise, $P$ selects just the top-left corner cell of the area, and the rest will follow as the chosen cut protocol preserves the cyclic order).
	\item Apply the first part of the chosen cut protocol to select a template of a pentomino with the same type as $B_i$.
	\item Apply the printing protocol on the selected template and the selected area. Continue the second part of the chosen cut protocol invoked in Step 1 to revert the grid into its original state.
	\item Reconstruct a template that has just been used and replenish the pile of templates with it. Let $V$ verify again that all 63 templates are correct (otherwise $V$ rejects). Note that $V$ does not know which template has been used.
\end{enumerate}

Finally, $P$ reveals all cards on the cells that contain a number (in the original Five Cell puzzle). $V$ verifies that the numbers on the cards match the numbers in the cells (otherwise $V$ rejects). $P$ also reveals all dummy cards that they are still blank. If all verification steps pass, then $V$ accepts.

This protocol uses $\Theta(mn)$ cards and $\Theta(mn)$ shuffles.

\subsection{Proof of Correctness and Security}
We will prove the perfect completeness, perfect soundness, and zero-knowledge properties of this protocol.

\begin{lemma}[Perfect Completeness] \label{lem1}
If $P$ knows a solution of the Five Cells puzzle, then $V$ always accepts.
\end{lemma}

\begin{proof}
Suppose $P$ knows an extended solution of the puzzle. Consider each $i$-th iteration of the main protocol.

\begin{itemize}
	\item In Step 3, since $B_1,B_2,...,B_k$ form a partition of the grid, $B_i$ does not overlap with $B_1,B_2,...,B_{i-1}$. Thus, the printing protocol will pass, and the numbers in $B_i$ will be printed on the grid.
	
	\item In Step 4, $P$ reconstructs a template that has just been used, so this step will pass.
\end{itemize}

Therefore, every iteration will pass. After $k$ iterations, all numbers in $P$'s extended solution will be printed on the grid, so all numbers in the original puzzle will match the numbers on the corresponding cards.

Hence, we can conclude that $V$ always accept.
\end{proof}

\begin{lemma}[Perfect Soundness] \label{lem2}
If $P$ does not know a solution of the Five Cells puzzle, then $V$ always rejects.
\end{lemma}

\begin{proof}
We will prove the contrapositive of this statement. Suppose $V$ accepts, which means the verification passes in all steps. Consider the main protocol.

Since Step 4 passes for every iteration, all 63 templates are correct after each iteration (and also at the beginning of the protocol), which implies the numbers printed in every iteration form a shape of a pentomino and follow the rule of the puzzle.

In Step 3, since the printing protocol passes for every iteration, $B_i$ does not overlap with $B_1,B_2,...,B_{i-1}$ for every $i$. Also, the combined area of $B_1,B_2,...,B_k$ is $5k=mn$, which implies $B_1,B_2,...,B_k$ must form a partition of the grid.

Since the final verification passes, all numbers in the original puzzle match the numbers on the corresponding cards.

Hence, we can conclude that the original puzzle grid is partitioned into pentominoes according to the rule, which means $P$ must know a valid solution of the puzzle.
\end{proof}

\begin{lemma}[Zero-Knowledge] \label{lem3}
During the verification, $V$ obtains no information about $P$'s solution.
\end{lemma}

\begin{proof}
We will prove that the interaction between $P$ and $V$ can be simulated by a simulator $S$ that does not know $P$'s solution. It is sufficient to show that all distributions of cards that are turned face-up can be simulated by $S$.

\begin{itemize}
	\item In Steps 3 and 6 of the chosen cut protocol in Section \ref{chosen}, the \mybox{1} has probability $1/q$ to be at each of the $q$ columns (due to the pile-shifting shuffles), so these two steps can be simulated by $S$.

	\item In Step 2(b) of the printing protocol in Section \ref{print}, there is only one deterministic pattern of the cards that are turned face-up (all blank cards), so this step can be simulated by $S$.
	
	\item In Step 4 of the main protocol, there is only one deterministic pattern of the cards that are turned face-up (all correct templates), so this step can be simulated by $S$.
\end{itemize}

Hence, we can conclude that $V$ obtains no information about $P$'s solution.
\end{proof}

\section{ZKP Protocol for Meadows} \label{meadows}
The key observation is that there are only $n$ different sizes of square in the solution of Meadows. Therefore, we need to construct only $n$ different templates.

In our protocol, $P$ constructs $n$ templates, one for each size of square. Each template has size $n \times n$, and the square is placed at the top-left corner of the template. A cell inside the square is represented by a \mybox{1}, while a cell outside the square is represented by a blank card (see Fig. \ref{figB}).

\begin{figure}
\centering
\begin{tikzpicture}
\node at (0,0) {\mybox{}};
\node at (0.5,0) {\mybox{}};
\node at (1,0) {\mybox{}};
\node at (1.5,0) {\mybox{}};
\node at (2,0) {\mybox{}};

\node at (0,0.6) {\mybox{}};
\node at (0.5,0.6) {\mybox{}};
\node at (1,0.6) {\mybox{}};
\node at (1.5,0.6) {\mybox{}};
\node at (2,0.6) {\mybox{}};

\node at (0,1.2) {\mybox{}};
\node at (0.5,1.2) {\mybox{}};
\node at (1,1.2) {\mybox{}};
\node at (1.5,1.2) {\mybox{}};
\node at (2,1.2) {\mybox{}};

\node at (0,1.8) {\mybox{1}};
\node at (0.5,1.8) {\mybox{1}};
\node at (1,1.8) {\mybox{}};
\node at (1.5,1.8) {\mybox{}};
\node at (2,1.8) {\mybox{}};

\node at (0,2.4) {\mybox{1}};
\node at (0.5,2.4) {\mybox{1}};
\node at (1,2.4) {\mybox{}};
\node at (1.5,2.4) {\mybox{}};
\node at (2,2.4) {\mybox{}};
\end{tikzpicture}
\hspace{1.5cm}
\begin{tikzpicture}
\node at (0,0) {\mybox{}};
\node at (0.5,0) {\mybox{}};
\node at (1,0) {\mybox{}};
\node at (1.5,0) {\mybox{}};
\node at (2,0) {\mybox{}};

\node at (0,0.6) {\mybox{1}};
\node at (0.5,0.6) {\mybox{1}};
\node at (1,0.6) {\mybox{1}};
\node at (1.5,0.6) {\mybox{1}};
\node at (2,0.6) {\mybox{}};

\node at (0,1.2) {\mybox{1}};
\node at (0.5,1.2) {\mybox{1}};
\node at (1,1.2) {\mybox{1}};
\node at (1.5,1.2) {\mybox{1}};
\node at (2,1.2) {\mybox{}};

\node at (0,1.8) {\mybox{1}};
\node at (0.5,1.8) {\mybox{1}};
\node at (1,1.8) {\mybox{1}};
\node at (1.5,1.8) {\mybox{1}};
\node at (2,1.8) {\mybox{}};

\node at (0,2.4) {\mybox{1}};
\node at (0.5,2.4) {\mybox{1}};
\node at (1,2.4) {\mybox{1}};
\node at (1.5,2.4) {\mybox{1}};
\node at (2,2.4) {\mybox{}};
\end{tikzpicture}
\caption{$5 \times 5$ templates of squares with side lengths 2 (left) and 4 (right)}
\label{figB}
\end{figure}
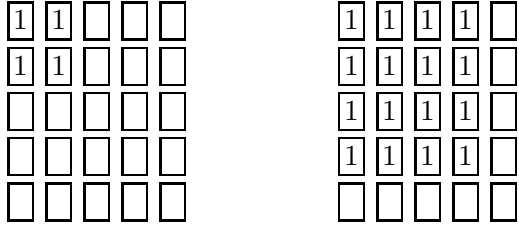

\subsection{Main Protocol}
Initially, $P$ publicly places a blank card on every cell in the grid. To handle edge cases, $P$ publicly appends $n-1$ rows and $n-1$ columns of ``dummy'' blank cards to the bottom and to the right of the grid. Then, $P$ turns all card face-down. We now have an $(2n-1) \times (2n-1)$ matrix of cards.

Also, $P$ constructs $n$ templates of all $n$ sizes of square and lets $V$ verify that all templates are correct (otherwise $V$ rejects).

Let $d_1,d_2,...,d_k$ be the dotted cells in the grid. Suppose that in $P$'s solution, the grid is partitioned into $k$ squares $B_1,B_2,...,B_k$, with each square $B_i$ containing cell $d_i$. For each $i=1,2,...,k$, $P$ performs the following steps.

\begin{enumerate}
	\item Apply the first part of the chosen cut protocol to select an $n \times n$ area whose top-left corner is the top-left corner of square $B_i$. (To be precise, $P$ selects just the top-left corner cell of the area, and the rest will follow).
	\item Apply the first part of the chosen cut protocol to select a template of a square with the same size as $B_i$.
	\item Apply the printing protocol on the selected template and the selected area. Continue the second part of the chosen cut protocol invoked in Step 1 to revert the grid into its original state.
	\item Turn over cards on cells $d_i,d_{i+1},...,d_k$ to show that the card on $d_i$ is a $\mybox{1}$, and the cards on $d_{i+1},d_{i+2},...,d_k$ are all blank cards (otherwise $V$ rejects).
	\item Reconstruct a template that has just been used and replenish the pile of templates with it. Let $V$ verify again that all $n$ templates are correct (otherwise $V$ rejects). Note that $V$ does not know which template has been used.
\end{enumerate}

Finally, $P$ reveals all cards in the grid that they are all $\mybox{1}s$ (otherwise $V$ rejects). $P$ also reveals all dummy cards that they are still blank. If all verification steps pass, then $V$ accepts.

This protocol uses $\Theta(n^3)$ cards and $\Theta(kn^2)$ shuffles.

\subsection{Proof of Correctness and Security}
We will prove the perfect completeness, perfect soundness, and zero-knowledge properties of this protocol.

\begin{lemma}[Perfect Completeness] \label{lem4}
If $P$ knows a solution of the Meadows puzzle, then $V$ always accepts.
\end{lemma}

\begin{proof}
Suppose $P$ knows a solution of the puzzle. Consider each $i$-th iteration of the main protocol.

\begin{itemize}
	\item In Step 3, since $B_1,B_2,...,B_k$ form a partition of the grid, $B_i$ does not overlap with $B_1,B_2,...,B_{i-1}$. Thus, the printing protocol will pass, and all \mybox{1}s in $B_i$ will be printed on the grid.
	
	\item In Step 4, since all \mybox{1}s inside $B_i$ have already been printed, and all \mybox{1}s inside $B_{i+1},B_{i+2},...,B_k$ have not yet been printed, this step will pass.
	
	\item In Step 5, $P$ reconstructs a template that has just been used, so this step will pass.
\end{itemize}

Therefore, every iteration will pass. After $k$ iterations, all cells in the grid will be printed with \mybox{1}s, so the final verification will also pass.

Hence, we can conclude that $V$ always accept.
\end{proof}

\begin{lemma}[Perfect Soundness] \label{lem5}
If $P$ does not know a solution of the Meadows puzzle, then $V$ always rejects.
\end{lemma}

\begin{proof}
We will prove the contrapositive of this statement. Suppose $V$ accepts, which means the verification passes in all steps. Consider the main protocol.

Since Step 5 passes for every iteration, all $n$ templates are correct after each iteration (and also at the beginning of the protocol), which implies the \mybox{1}s printed in every iteration form a shape of a square.

Since Step 4 passes for every iteration, the square printed in each $i$-th iteration must contain exactly one dotted cell, which is $d_i$.

In Step 3, since the printing protocol passes for every iteration, $B_i$ does not overlap with $B_1,B_2,...,B_{i-1}$ for every $i$. Also, since the final verification passes, the combined area of $B_1,B_2,...,B_k$ must cover the whole grid, i.e. $B_1,B_2,...,B_k$ form a partition of the grid.

Hence, we can conclude that the puzzle grid is partitioned into squares, with each one contaning exactly one dotted cell, which means $P$ must know a valid solution of the puzzle.
\end{proof}

\begin{lemma}[Zero-Knowledge] \label{lem6}
During the verification, $V$ obtains no information about $P$'s solution.
\end{lemma}

\begin{proof}
We will prove that the interaction between $P$ and $V$ can be simulated by a simulator $S$ that does not know $P$'s solution. It is sufficient to show that all distributions of cards that are turned face-up can be simulated by $S$.

The zero-knowledge property of the chosen cut protocol and the printing protocol has been proved in the proof of Lemma \ref{lem3}, so it is sufficient to consider only the main protocol.

\begin{itemize}
	\item In Step 4, there is only one deterministic pattern of the cards that are turned face-up, so this step can be simulated by $S$.
	
	\item In Step 5, there is only one deterministic pattern of the cards that are turned face-up (all correct templates), so this step can be simulated by $S$.
\end{itemize}

Hence, we can conclude that $V$ obtains no information about $P$'s solution.
\end{proof}

\section{Future Work}
We constructed the printing protocol, which can be used to develop ZKP protocols for decompositon puzzles.\footnote{The printing protocol can also be used to construct a ZKP for Shikaku. However, for an $m \times n$ Shikaku grid, we need to prepare $mn$ templates, each having size $m \times n$, resulting in the total of $\Theta(m^2n^2)$ cards, while the existing ZKP protocol \cite{shikaku} uses only $\Theta(mn)$ cards.} However, the limitation of this protocol is that the number of different possible types of region (which is equal to the number of templates we need to prepare) must be polynomially bounded. A possible future work is to find a technique to deal with decompositon puzzles which the number of different possible types of region are exponentially bounded, e.g. Fillomino.

\subsubsection*{Acknowledgement}
The authors would like to thank the anonymous reviewers of LATINCRYPT 2023 who kindly suggested the main idea of the printing protocol. Without them, this paper would not have been possible. The authors would also like to thank Kyosuke Hatsugai, Yoshiki Abe, and Tomoki Ono for a valuable discussion on this research. This work was supported by JSPS KAKENHI Grant Numbers JP23H00468, JP23H00479, JP23K17455, JP22H03590, JP21H03395, and JST CREST JPMJCR23M2.

\end{document}